\documentclass{techrep}
 
\usepackage{amsfonts}
\usepackage{amsthm}

\newcommand{\comment}[1]{}

\newtheorem{theorem}{Theorem}{\bfseries}{\itshape}

\newtheorem{lemma}{Lemma}{\bfseries}{\itshape}

{\bfseries}{\itshape}

{\bfseries}{\itshape}

{\bfseries}{\itshape}

{\bfseries}{\itshape}

\newtheorem{assumption}{Assumption}{\bfseries}{\rm}

\theoremstyle{definition}

{\bfseries}{\rm}

\theoremstyle{plain}

\usepackage{verbatim}
\usepackage[utf8]{inputenc}
\usepackage{microtype}
\usepackage{amsthm,amsmath,graphicx}
\usepackage{caption}
\usepackage{subcaption}

\usepackage[letterpaper]{hyperref}
\usepackage[table,dvipsnames]{xcolor}
\definecolor{linkblue}{named}{Blue}
\hypersetup{colorlinks=true, linkcolor=linkblue,  anchorcolor=linkblue,
	citecolor=linkblue, filecolor=linkblue, menucolor=linkblue,
	urlcolor=linkblue} 
\usepackage{wasysym}
\usepackage{graphicx}
\usepackage{caption}
\usepackage{enumitem}
\usepackage{thmtools, thm-restate}
\usepackage{wrapfig}
\usepackage{float}
\usepackage{stackrel}
\usepackage{fixfoot}
\usepackage[utf8]{inputenc}
\usepackage[english]{babel}
\usepackage{mathrsfs}

\title{Gathering by Repulsion}

\author{Prosenjit Bose\footnote{School of Computer Science, Carleton University, Canada. email: {jit@scs.carleton.ca}. Research supported in part by NSERC.} \and \and Thomas C. Shermer\footnote{School of Computing Science, Simon Fraser University, Canada. email:{shermer@sfu.ca}}}

\begin{document}


\newtheorem{thm}{Theorem}
\newtheorem{lem}[thm]{Lemma}

\newcommand{\notes}[1]{\marginpar{\tiny{#1}}}

\newlength\problemsep
\setlength\problemsep{10pt}
\newcommand{\given}{}
\newcommand{\find}{}
\newcommand{\nogo}{\textsc{impossible}}
\newcommand{\go}{\textsc{possible}}
\newcommand{\bdP}{\ensuremath{\partial P}}
\newcommand{\bdPccw}[1]{\ensuremath{\partial^{+}(#1)}}
\newcommand{\bdPcw}[1]{\ensuremath{\partial^{-}(#1)}}
\newcommand{\ERCW}[1]{\ensuremath{R_{\textsl{cw}}( #1 )}}
\newcommand{\ERCCW}[1]{\ensuremath{R_{\textsl{ccw}}( #1 )}}
\newcommand{\ERSPLIT}[1]{\ensuremath{S( #1 )}}
\newcommand{\vangle}[1]{\ensuremath{u_{#1}}}
\newcommand{\vanglecw}[1]{\ensuremath{u_{#1}^{\textsl{\tiny cw}}}}
\newcommand{\vangleccw}[1]{\ensuremath{u_{#1}^{\textsl{\tiny ccw}}}}
\newcommand{\accumccw}[1]{\ensuremath{\textsl{accum}^{\textsl{\tiny ccw}}(#1)}}
\newcommand{\accumcw}[1]{\ensuremath{\textsl{accum}^{\textsl{\tiny cw}}(#1)}}
\newcommand{\pifrac}[2]{\ensuremath{\frac{#1 \pi}{#2}}}
\newcommand{\piover}[1]{\pifrac{}{#1}}
\newcommand{\piovertwo}{\ensuremath{\frac{\pi}{2}}}
\newcommand{\threepiovertwo}{\ensuremath{\frac{3\pi}{2}}}
\newcommand{\tqpi}{\pifrac{3}{4}}
\newcommand{\sed}[1]{\ensuremath{D(#1)}}
\newcommand{\bd}[1]{\ensuremath{\partial D(#1)}}

\newenvironment{problem}[1] {			
	\vspace{\problemsep}
	\noindent{\scshape \underline{#1}}
	\renewcommand\descriptionlabel[1]%
		{\hspace{\labelsep}\textsc{##1}}
	\renewcommand{\given}{\item[Given:]}
	\renewcommand{\find}{\item[Find:]}
	\begin{description}
}{
	\end{description}
	\vspace{\problemsep}
}

\maketitle
\begin{abstract}
We consider a repulsion actuator located in an $n$-sided convex environment full of point particles. When the actuator is
activated, all the particles move away from the actuator.
We study the problem of gathering all the particles to a point. 
We give an $O(n^2)$ time algorithm to compute all the actuator locations that gather the particles to one point with one activation, and an $O(n)$ time algorithm to find a single such actuator location if one exists. We then provide an $O(n)$ time algorithm to place the optimal number of actuators whose sequential activation results in the gathering of the particles when such a placement exists. 

\end{abstract}

\section{Introduction}
	In this paper, we consider some basic questions about \emph{movement by
	repulsion}. Here a point actuator repels particles, or put another way, particles move so as to locally maximize their distance from the actuator.
	This problem models magnetic repulsion,
	movement of floating objects due to waves,
	robot movement (if robots are programmed to move away from certain stimuli),
	and crowd movement in an emergency or panic situation.
	It is, in one sense, the opposite of movement by attraction, which
	has recently been an active topic of research \cite{biro, beacon1, beacon2, beacon3, beacon4, beacon5, beacon6, beacon7}.
	
	\subsection{Related work}
	We initiate the study of repulsion in polygonal settings. 
	The closest comparable work is the work on attraction. Although attraction and repulsion have a similar definition, each has a distinct character. 
      Attraction as it has been studied is mainly a two-point relation:
       a point $p$ attracts a point $q$ if $q$, moving locally to minimize distance to $p$, eventually reaches $p$.
	  In repulsion, $p$ cannot repulse $q$ to itself; $p$ must always repulse $q$ to some other point $r$. 
      Thus repulsion is a three-point relation.

In attraction, if a particle is attracted onto an edge by a beacon, it is pulled towards the point $p$ where
     there is a perpendicular from the beacon to the line through the edge.   If $p$ is on the edge, this creates a stable minimum at $p$, and particles accumulate at such mimima.  As well, particles can accumulate on some convex vertices.
     
     In repulsion, if a particle is repelled onto an edge by a repulsion actuator, it is pushed away from the point $p$ with the perpendicular to the actuator.  This implies that $p$ is an \emph{unstable} maximum.  We forbid particles from stopping at unstable maxima, so in repulsion the only accumulation points will be convex vertices. We elaborate further on our model in Subsection \ref{sec:model}.
	
	In this article, we highlight some of the similarities as well as distinctions between these two concepts. 
	For instance, Biro \cite{biro} designed an $O(n^2)$ time algorithm for computing the \emph{attraction kernel} of a simple $n$-vertex polygon $P$; 
     these are all points $p \in P$ that attract all points $q \in P$. The closest counterpart of this for repulsion, 
     which we call the \emph{repulsion kernel} of a polygon $P$, is all points $p \in P$ such that there exists a point
     $r \in P$ such that $p$ repels all points in $P$ to $r$.  We give an $O(n^2)$ time algorithm to compute the repulsion kernel
     of an $n$-vertex convex polygon, and an $O(n)$ time algorithm to find a single-point in the repulsion kernel or report that the kernel is empty.
    
    Both the attraction kernel and the repulsion kernel are concerned with the problem of gathering particles to a point. When the repulsion kernel is empty, it may be the case that we can still gather all particles to a point  using more than one repulsion actuator. In this vein, we prove that this is impossible in a polygon with three acute angles. In a convex polygon with at most two acute angles, two repulsion actuators are always sufficient and sometimes necessary. We then provide an $O(n)$ time algorithm to place the optimal number of actuators.
   
	\subsection{The model}\label{sec:model}
	
	We start with an $n$-vertex convex polygon $P$, which includes its interior.
	Before the activation of any repulsion actuator, there is a particle on every point of the polygon, including the boundary.
	During and after activation, we allow many particles to be on the same point; once two particles reach the same
	point, they travel identically, so we consider them to be one particle.
	
	We restrict the location of the repulsion actuator to points in $P$;
	allowing the actuator to reside outside $P$ leads to a variation of the problem
	in which convex polygons are easily dispensed. 
	
	See Figure \ref{fig:Definitions} for an illustration of the following definitions.
	The activation of an actuator will
	cause all particles to move to locally maximize their distance from the actuator.
	This means that if a particle is in the interior of $P$, then it moves in a
	straight line away from the actuator's location.
	If a particle is on an edge of the polygon,
	then it proceeds along the edge in the direction that will further its
	distance from the active actuator.
	Once moving, a particle moves until it is stable and can no longer locally
	increase its distance from the actuator.
	Stable maxima happen at vertices where neither of the two edges allows movement
	away from the actuator.
	We call such vertices the \emph{accumulation points} of the activation.

	\begin{figure}
		\centering 
		\includegraphics[scale=0.8]{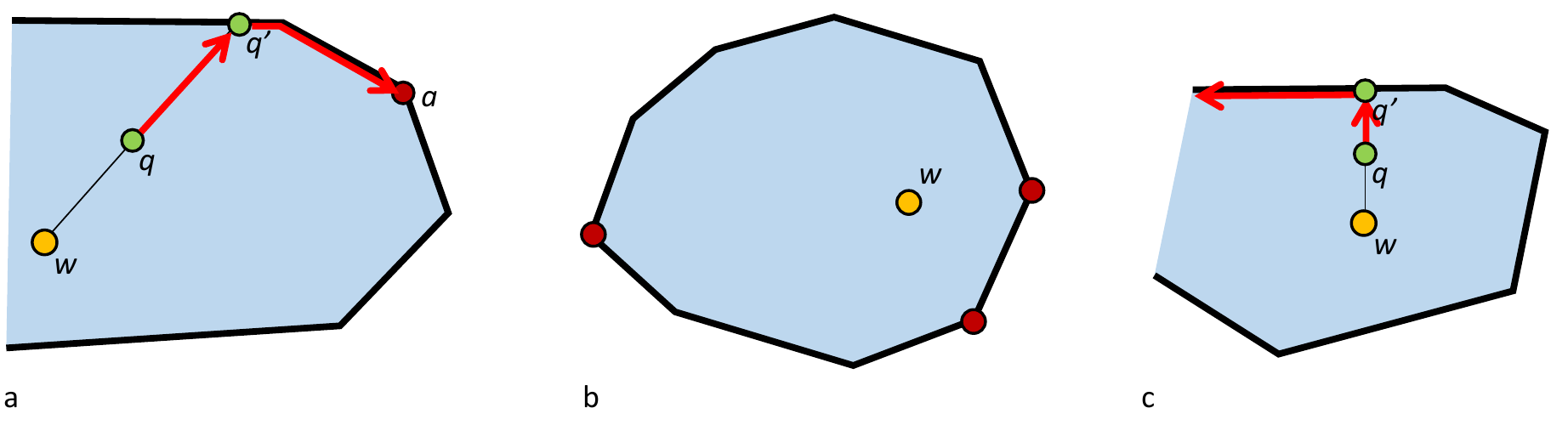}
		\caption{  
			(a) An activation at $w$ drives the particle at $q$ away from $w$.
			On reaching an edge at $q'$,
			it will continue to move away from $w$,
			until it reaches a local maximum of distance from $w$ at $a$.
			(b) Accumulation points of an active actuator at $w$.
			(c) At an unstable maximum, such as $q'$, particles will turn left.}
		\label{fig:Definitions} 
	\end{figure}

	Unstable maxima happen when a particle is on an edge
	where one or both directions give no differential change of distance from the actuator;
	this happens only at the perpendicular projection of the actuator onto the edge
	(see Figure \ref{fig:Definitions}c).
	A particle at an unstable maximum will move off of it in a direction of no improvement and then will
	be able to increase the distance from the actuator by
	continuing in that direction.
	To maintain a deterministic model, we will assume that particles
	move counterclockwise around the polygon at unstable maxima if there is a choice of two directions of no improvement.  However, the choice of counterclockwise motion is
	arbitrary, and does not affect our results.
	
	We may activate actuators sequentially from several places inside the polygon. We would like for every activation of an actuator to be from a location without
	particles, but the particle-on-every-point model forbids this on the first activation.
	So, when we choose a location for the first actuator, we remove the
	particle at that location from the problem.  For subsequent activations, however, we
	do require that the actuator's position be chosen from the points of
	the polygon without particles.
	
The main question we consider is when can we place a sequence of points such that repulsion from those points gathers all other points in the polygon to one point? When the repulsion kernel is non-empty, one point is sufficient.
 In general, our goal is to minimize the number of sequential activations performed to gather all the particles to one point.
	If all the particles in a polygon can be gathered to a point with $k$ sequential activations of actuators, we call the polygon \emph{$k$-gatherable}.  If this is not possible for any $k$, then
	we call the polygon \emph{ungatherable}.

\section{Background, notation, and terminology}

	\subsection{General notation}
		We will use the convention that the vertices of $P$ are $v_0, v_1, \ldots ,
		v_{n-1}$ in counterclockwise order around the polygon.
		Vertex indices are taken modulo $n$, so 
		$ v_{-1} = v_{n-1},
		v_0 = v_n,
		v_1 = v_{n+1}$, etc.
		Edges are denoted $e_0, e_1, \ldots e_{n-1},$ with $e_i$ being the edge
		between $v_i$ and $v_{i+1}$.
		The boundary of the polygon $P$ will be denoted \bdP, and by
		$\bdP(p, q)$ we mean the part of $\bdP$ from $p$ counterclockwise to
		$q$.
		In reference to curves, line segments, or intervals, we use the usual
		parentheses to denote relatively open ends and square brackets
		to denote relatively closed ends.
		Thus $\bdP[p, q)$ is the boundary from $p$ to $q$, including $p$ but not $q$.
		Given three distinct points $a, b, c$ in the plane, by $\angle{abc}$ we mean the counterclockwise angle between 
		the ray from $b$ to $a$ and the ray from $b$ to $c$.

	\subsection{Slabs and the three regions of an edge}
  		Consider a polygon edge with particles covering it. When an actuator is activated, 
  		depending on its location relative to the edge, 
  		there are three possible effects on the particles:  it drives them
  		counterclockwise over the entire edge, it drives them clockwise over
  		the entire edge, or it drives some of them clockwise and some of them
  		counterclockwise (see Figure \ref{fig:WoofEffects}).  In the latter case,
  		a perpendicular from the edge to the actuator exists, and the particles clockwise
  		of the perpendicular are driven clockwise, and the particles counterclockwise
  		of the perpendicular are driven counterclockwise.
  		The point where the perpendicular hits the edge is called a \emph{split
  		point}.  We allow split points at the endpoint of an edge if a perpendicular
  		from the endpoint to the actuator exists.
  		
  		\begin{figure}
			\centering 
			\includegraphics[scale=0.8]{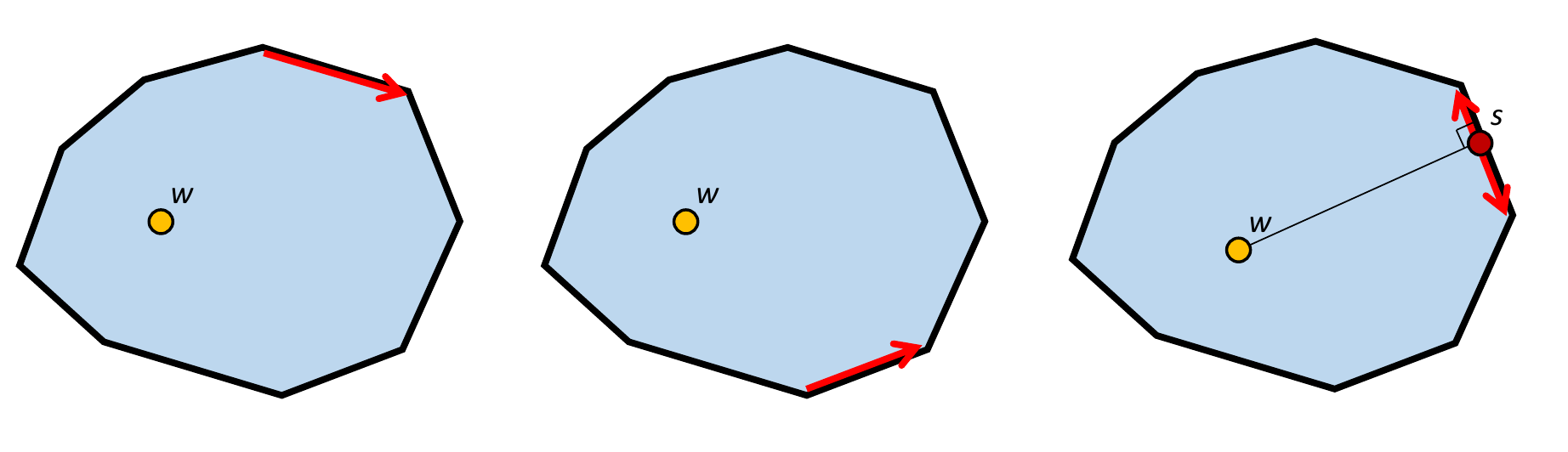}
			\caption{We use arrows in the diagrams to show the direction that the
					particles are driven.  
				(a) The activation drives the particles (on the indicated edge) clockwise.
				(b) The activation drives the particles counterclockwise.
				(c) The activation splits the particles at $s$, driving some clockwise and some
				counterclockwise.}
			\label{fig:WoofEffects} 
		\end{figure}
		
		We divide the inner halfplane of an edge $e$ into three regions
		depending on what effect an activation in the region has on the particles on the edge.
		This is done by drawing interior-facing perpendiculars to the edge at each of
		its vertices.
		The regions are \ERCW{e}, where an activation drives the particles clockwise,
		\ERCCW{e}, where an activation drives the particles counterclockwise,
		and \ERSPLIT{e}, where an activation drives some particles clockwise and some
		counterclockwise.  We refer
		to \ERSPLIT{e} as the \emph{slab} of $e$.  The slab is closed on its
		boundaries, and \ERCW{e} and \ERCCW{e} are open where they meet \ERSPLIT{e}.
	
  		\begin{figure}
	 		\centering 
	 		\includegraphics[scale=0.8]{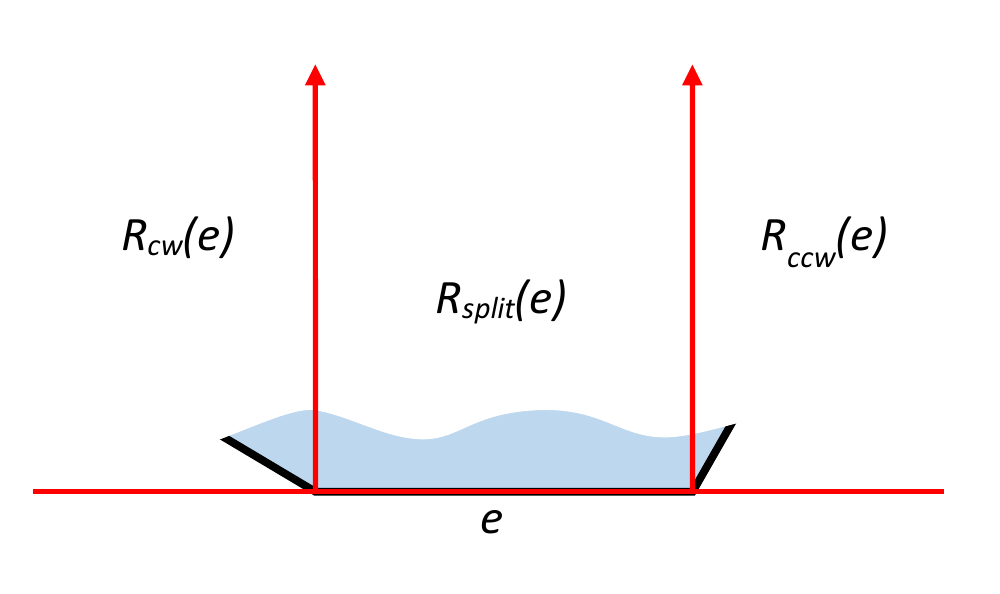}
			\caption{The three regions \ERCW{e}, \ERCCW{e}, and \ERSPLIT{e}.}
			\label{fig:EdgeRegions} 
		\end{figure}

	\subsection{Flow diagrams}

		Given a polygon $P$ and a location $w$ of an actuator, we may find the accumulation
		points and the split points, and mark each edge (or portion of a split edge)
		with the direction of particle movement along that edge, as in Figure
		\ref{fig:AccumSplit}.  We call a diagram of this a \emph{flow diagram for $w$ with respect to $P$}.
		
		\begin{figure}
	 		\centering 
	 		\includegraphics[scale=0.8]{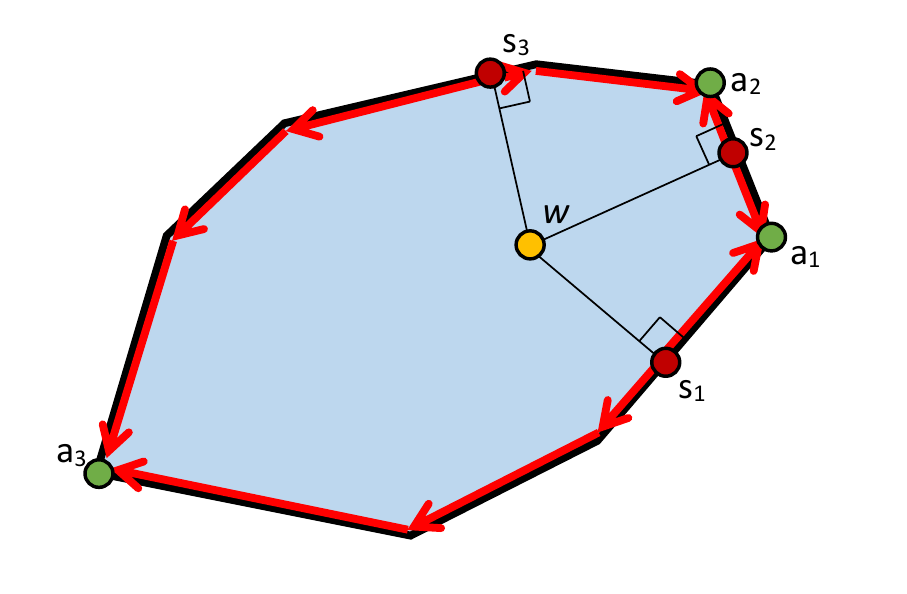}
			\caption{A flow diagram, showing the direction of particle movement, along with
			the accumulation points and split points, given an actuator at $w$.}
			\label{fig:AccumSplit} 
		\end{figure}
		
		\begin{lemma}\label{lem:alternating}
		    In a traversal of \bdP, accumulation and split points alternate.
		\end{lemma}		
		
		\begin{proof}		
		    Note that in a flow diagram the only points of the boundary with two
		    opposing directions of particle movement are the accumulation points, where
		    the movement is towards the point, and the split points, where the
		    movement is away from the point.  Thus, between any two consecutive split
		    points on the boundary, there must be an accumulation point, and between
		    any two consecutive accumulation points, there must be a split point. 
		    This implies the lemma.
		\end{proof}

		\begin{theorem}\label{thm:oneSlab}
		    A convex polygon $P$ is $1$-gatherable from $w$ iff $w$ lies in the slab of exactly
		    one of the edges of $P$.
		\end{theorem}
		
		\begin{proof}
			A polygon is $1$-gatherable from $w$ iff an actuator at $w$ has one accumulation point.
			Since accumulation and split points alternate, this holds iff the actuator
			has exactly one split point.  Since an actuator has a single split point in every
			slab that it is in (and no others), the result follows.
		\end{proof}
		
		The boundary of the slab for edge $e$ consists of $e$ and two rays perpendicular to $e$. If we produce these two rays for each edge of $P$, and intersect all these rays with $P$, we get a set of at most $2n$ chords that define a decomposition that we call the {\em slab decomposition} of $P$. An example slab decomposition is shown in Figure \ref{fig:BoundaryWoof}. The cells of this decomposition have the property that if two points are in a cell, then these two points are in exactly the same set of slabs of $P$.
		
		Theorem \ref{thm:oneSlab} then immediately implies that the repulsion kernel of $P$ is the union of zero or more cells of the slab decomposition of $P$. This gives us the basis for an $O(n^2)$ time algorithm for finding the repulsion kernel. We start by constructing the slab decomposition. We can use topological sweep to compute a quad-edge data structure for the slab decomposition in $O(n^2)$ time \cite{sweep1, sweep2, quad}.
		
		\begin{theorem}
	The repulsion kernel of a convex polygon can be computed in $O(n^2)$ time.
		\end{theorem}
		\begin{proof}
		We construct the slab decomposition. As we construct the decomposition, we augment each edge with information about which slab or slabs it borders and to which side of the edge said slabs are on. (An edge may border two slabs if the two slabs each have a defining ray that are collinear.). Choose an arbitrary cell $c$ of the decomposition and determine how many slabs it is in. From this cell, perform a graph search on the dual of the decomposition. Each time we step over an edge, from one cell to another, during this search, we update in constant time the number slabs we are in, according to the information on the edge. We maintain a list of all cells where this value is one. At the end of the search, this list is the repulsion kernel. 
		\end{proof}

If we allow actuators to be located outside a polygon $P$, then every convex polygon is 1-gatherable.

    	\begin{lemma}\label{obs:1-gather}
 Every convex polygon is 1-gatherable from some point in the plane.
    \end{lemma}
    \begin{proof}
    If you go far enough away, you can always find a point that is not covered by any slab. For this point, there is only one accumulation point. Therefore, an activation of an actuator from this point moves all the particles to the accumulation point. 
    \end{proof}
   		
	Given the above, one may be tempted to believe that \emph{every} convex polygon is $1$-gatherable when the actuators are restricted to be inside the polygon.  However, this is not always the case.
	
	\begin{lemma}\label{lem:needTwoWoofs}
		For $k \geq 2$, the regular $(2k+1)$-gon $P_{2k+1}$ is not $1$-gatherable.
	\end{lemma}
	
	\begin{proof}
		Assume that the edge length of $P_{2k+1}$ is $2$, and that $e_0$ is oriented
		with direction $0$ (horizontal on the bottom of the polygon).  This is
		illustrated in Figure \ref{fig:Pentagon}\ for $P_5$.
		
		By Lemma \ref{lem:boundaryWoof}, we need only show that $P_{2k+1}$ is not 
		$1$-gatherable from its boundary.  By symmetry, we need consider only $e_{k+1}$.
		The edge $e_{k+1}$ starts at the top center of the polygon and proceeds
		downward to the left.  The slab $S(e_0)$ contains the upper half of $e_{k+1}$,
		as the distance $c$ (see figure) is greater than $1$.  (It is $1/\sin \alpha$,
		to be precise, where $\alpha$ is half the vertex angle, or $\frac{(2k-1)\pi}{4k+2}$.)
		Similarly, the slab $S(e_1)$ contains the bottom half of $e_{k+1}$.
		
		Thus, each point of $e_{k+1}$ is in $S(e_{k+1})$ and either $S(e_0)$ or
		$S(e_1)$ or both.  Thus, by Theorem \ref{thm:oneSlab}, the polygon is not
		$1$-gatherable from any point of $e_{k+1}$.
		
		The vertices are sometimes special cases, but here the vertex $v_{k+1}$ (the
		top vertex of the polygon) is in $S(e_0)$, $S(e_k)$, and $S(e_{k+1})$, and
		thus the polygon is not $1$-gatherable from there.  By symmetry, it is not
		$1$-gatherable from any vertex.

		\begin{figure}
 			\centering 
 			\includegraphics[scale=0.6]{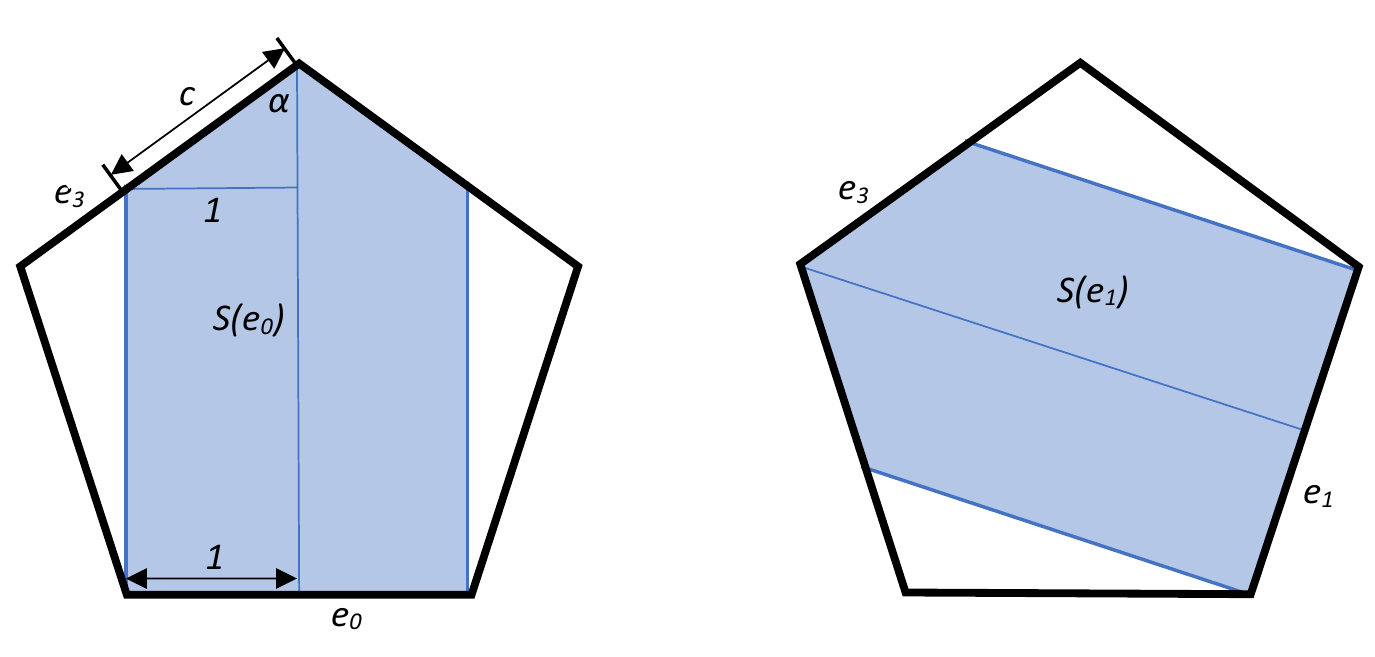}
			\caption{(a) $S(e_0)$ covers the top half of $e_3$.
			         (b) $S(e_1)$ covers the bottom half.}
			\label{fig:Pentagon} 
		\end{figure}
	\end{proof}
	
	In fact, some convex polygons may be ungatherable. It turns out that acute angles are a major impediment to gathering.
  		
    	\begin{lemma}\label{obs:acute}
     A particle that is at an acute vertex $v$ of $P$ cannot be moved by an actuator activated at any point in $P\setminus{v}$.
   	\end{lemma}
   	\begin{proof}
   	Given any point $p\in P\setminus{v}$, the acute vertex $v$ is a local maximum with respect to distance since any point in $P$ that is infinitesimially close to $v$ is closer to $p$ than $v$. 
   	\end{proof}
   		
   	This immediately implies the following.
   	
      \begin{theorem}
      A convex polygon with three acute vertices is not $k$-gatherable for any $k>0$.
      \end{theorem}
      
    For the remainder of the paper, we only consider convex polygons with at most two acute vertices.
      
\section{1-Gatherability}

We have shown so far that not all convex polygons are 1-gatherable. We have also given a complete characterization of when a convex polygon is 1-gatherable by computing the repulsion kernel of a polygon in $O(n^2)$ time. This begs the question whether it is possible to find a point from which the polygon is 1-gatherable more efficiently, without having to compute the repulsion kernel. We answer this question in the affirmative by providing an $O(n)$ time algorithm. Before presenting the algorithm, we highlight some useful geometric properties.

		\begin{lemma}\label{lem:perpendicularSupport}
			Let $a$ be an accumulation point of an actuator activated at $w$ in $P$.  The line $L$ that goes
			through $a$ and is perpendicular to $wa$ is a line of support of the polygon.
		\end{lemma}

		\begin{proof}
			Since $a$ is an accumulation point, it is a local maximum of distance from
			$w$.  Thus, the circle $C$ with center $w$ and radius $aw$ encloses the
			polygon in the neighborhood of $a$.  The line $L$ is tangent to (outside of) $C$ at
			$a$ and thus locally supports the polygon at $a$.  Since the polygon is
			convex, $L$ also globally supports the polygon.  
		\end{proof}

We now show that we can restrict our attention to particles starting only on the
		boundary of $P$.

		\begin{lemma} \label{lem:justboundary}
			An actuator in $P$ that 1-gathers all the particles on $\bdP$ also 1-gathers all particles in $P$.
		\end{lemma} 
		\begin{proof}
			The activation of an actuator in $P$ forces a particle $p$ in the interior of $P$ to move directly
			away from the actuator until it hits the boundary at some point $b$.
			Since there was a particle $p'$ whose initial position is $b$, 
			the particle $p$ will follow the path
			of $p'$ and stop at the same place $p'$ stops.  Thus, 
			the location of $p$ will always be accounted for by the position of $p'$.
			In other words,	$p$ is redundant and can be removed from the problem.  
		\end{proof}
		
		We can take this a step further and show that particles located on the interior of edges are redundant.
		
		\begin{lemma} \label{lem:justvertices}
			An actuator in $P$ that 1-gathers all the particles on the vertices $P$ also 1-gathers all particles on $\bdP$.
		\end{lemma} 
	
		\begin{proof}
		
		The activation of an actuator in $P$ forces a particle $p$ in the interior of an edge of $P$ to move along 
			 along the edge until it reaches a vertex $v$.  There was a
			particle $p'$ that started at $v$, and we can follow the proof of Lemma
			\ref{lem:justboundary}.
		\end{proof}
	
		The above lemmas show that particle movement can be restricted to the boundary. In fact, to solve the general problem, we only need to consider the problem where particles are only on vertices. We show a relationship between self-approaching paths and the path on the boundary followed by a particle under the influence of an actuator.  Recall that a directed path $\Pi$ is self-approaching if for any three consecutive points $x,y,z$ on the path, we have the property that $|xz| \geq |yz|$ \cite{icking}. 
		
		\begin{lemma} \label{lem:selfapproach}
		If $\bdP(x,y)$ is self-approaching from $x$ to $y$ then activating an actuator at $y$ sends all the particles on $\bdP(x,y)$ to $x$ along the boundary.
         \end{lemma}		
         
         \begin{proof}
       Let $z$ be an arbitrary point on $\bdP(x,y)$. We observed that activating an actuator at $y$ will move $z$ along the boundary. We need to establish in which direction the particle will move. Since $\bdP(x,y)$ is self-approaching from $x$ to $y$, we have that $|yz| \leq |yx|$. Therefore, the particle $z$ will move to $x$ since particles move in a direction to increase their distance from an actuator.
         \end{proof}
 
Next, we show that if the repulsion kernel is not empty, then there is at least one point on the boundary that is in the repulsion kernel. 

	\begin{lemma}\label{lem:boundaryWoof}
		Let $P$ be a convex polygon that is $1$-gatherable from a point $w$ in the interior of $P$, and let $a$ be
		the accumulation point for $w$.  Let $R$ be the ray from $a$ through $w$, not
		including the point $a$.  Then $P$ is $1$-gatherable from the point $w' = R \cap
		\bdP$, with $a$ as its accumulation point (see Figure
		\ref{fig:BoundaryWoof}).
	\end{lemma}

	\begin{figure}
 		\centering 
 		\includegraphics[scale=0.6]{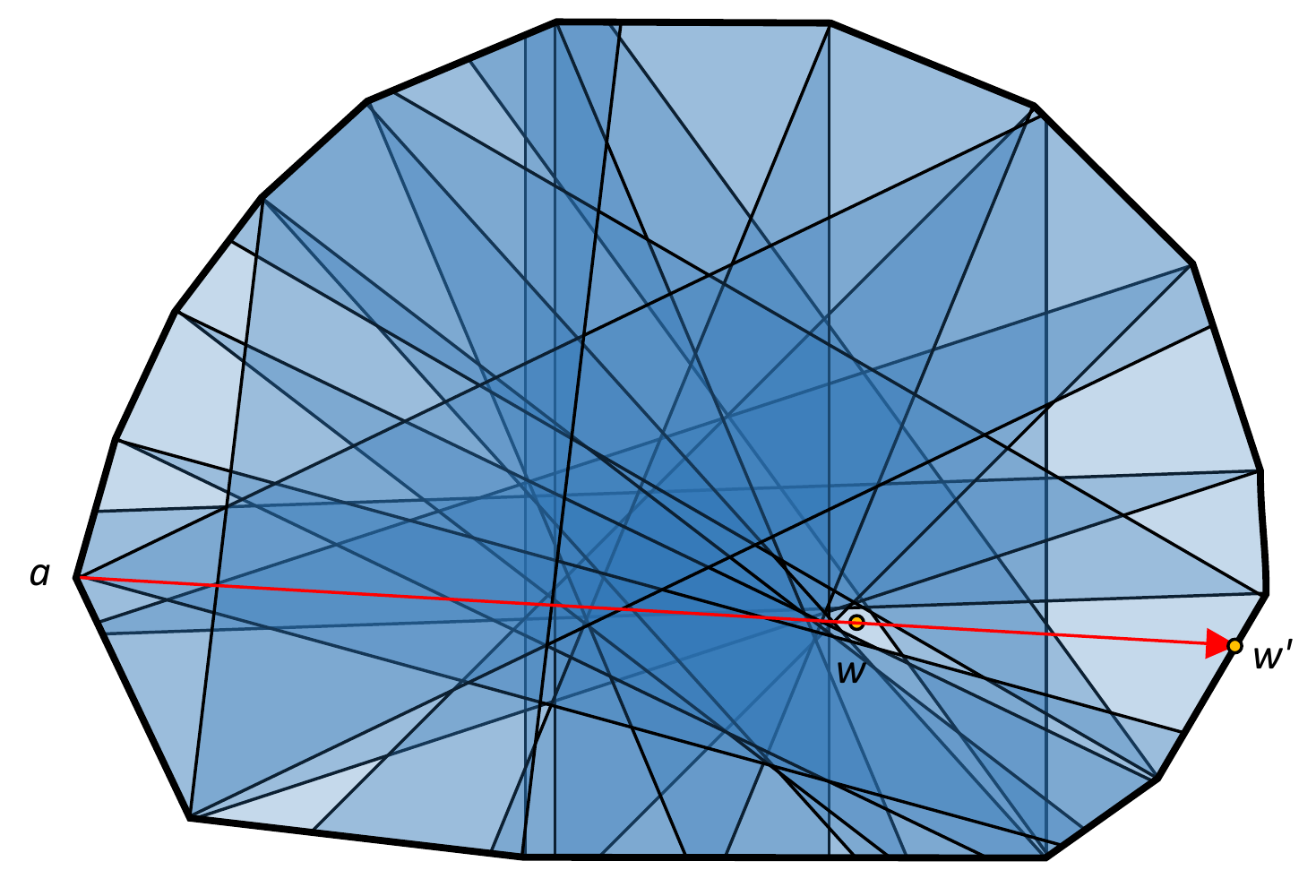}
		\caption{$P$ is 1-gatherable from $w'$. Each slab in $P$ is shown, with areas
		darkness corresponding to the number of slabs overlapping there.}
		\label{fig:BoundaryWoof} 
	\end{figure}

	\begin{proof}
		By Theorem \ref{thm:oneSlab}, the point of gatherability $w$ is in one edge
		$e$'s perpendicular slab.
		Without loss of generality, we assume that $e$ is horizontal at or below $w$
		(by rotation), that $a$ is not to the right of $w$ (by reflection), and that
		$e$ is $e_0 = v_0v_1$(by labelling).  Let $m$ be such that $a = v_m$.
		See Figure \ref{fig:BrainLabels}.
		Let $p$ be the point on $e$ which has a perpendicular through $w$. Note that $p$ is a split point for $w$.

		\begin{figure}
 			\centering 
			\includegraphics[scale=0.6]{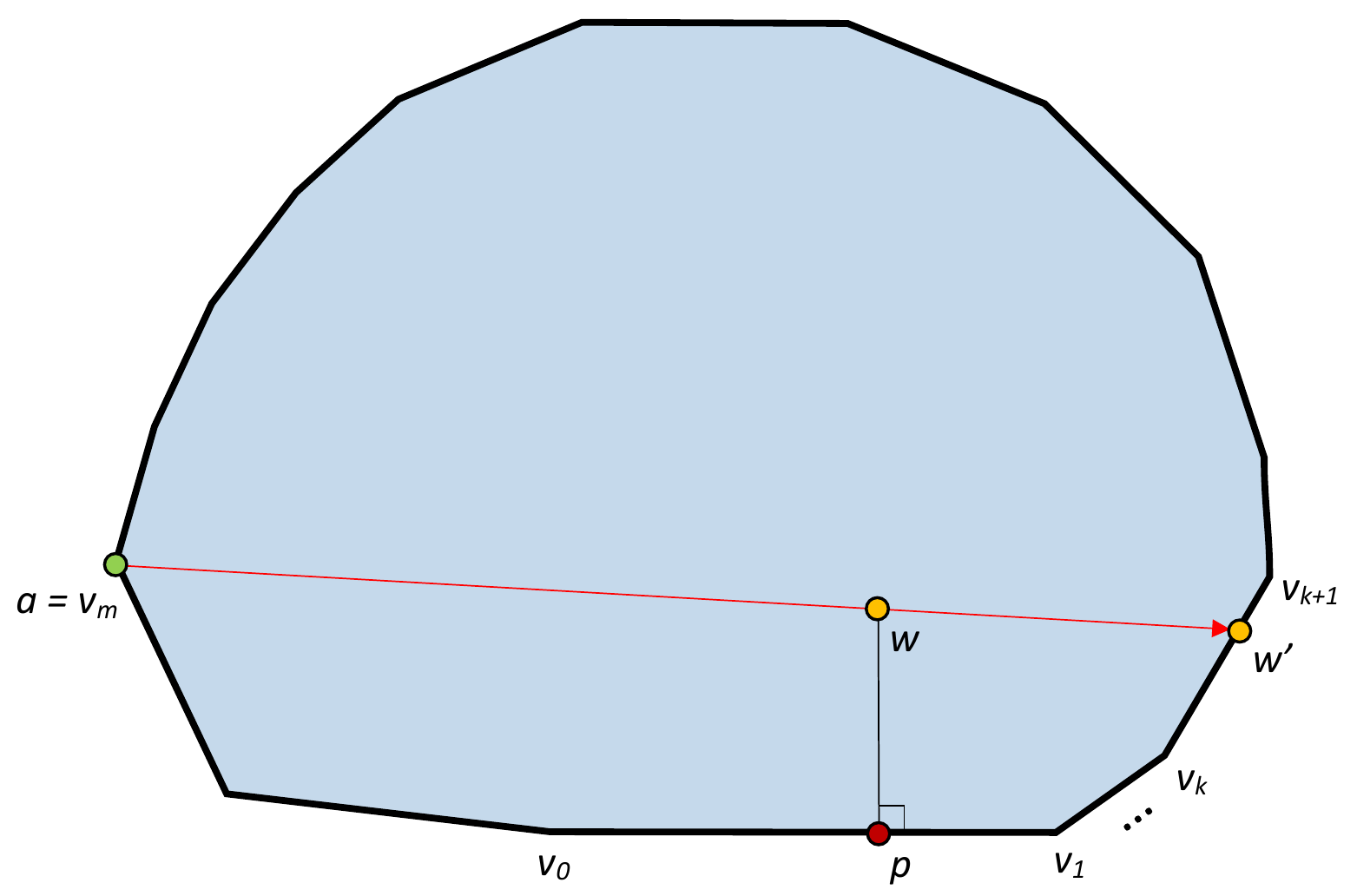}
			\caption{Some relevant points on the polygon.}
			\label{fig:BrainLabels} 
		\end{figure}

		To show that $P$ is 1-gatherable from the point $w'$ on $\bdP$, by Lemma \ref{lem:justvertices}, it suffices to show that the particles located on the vertices of $P$ move to one accumulation point with the activation of an actuator at $w'$. We will show that this accumulation point is $a$. We assume without loss of generality that $w'$ is located on the edge $e_k = [v_kv_{k+1})$. Recall that if $w'$ happens to be on $v_k$, then the placement of the actuator on $w'$ means the particle located at $w'$ is removed from consideration.
		
		We begin with the claim that an accumulation point for $w'$ is $a$.
		If this were not the case, then there would be a way to increase the
		distance from $w'$ on the boundary in the neighborhood of $a$. By Lemma \ref{lem:perpendicularSupport}, there is a line $L$ perpendicular to $wa$ that is a line of support of $P$ at $a$. By construction, $L$ is perpendicular to $w'a$. Therefore, $a$ is a local maximum with respect to $w'$, and thus is an accumulation point for $w'$.
		 We will now show that particles located at all other vertices move to $a$ when an actuator is activated at $w'$.

		Since $p$ is a split point for $w$, we have that upon activation of $w$, the particles on the vertices on $\bdP(a,p)$ move
		clockwise along the boundary to $a$. Similarly, the particles on the vertices on $\bdP(p,a)$ move counterclockwise along the boundary to $a$. By Theorem \ref{thm:oneSlab}, this means that $w$ is in all of the regions
		$\ERCCW{e_1}$, $\ERCCW{e_1}, \ldots, \ERCCW{e_{m-1}}$ and $w$ is in
		$\ERCW{e_m}$, $\ERCW{e_{m+1}}, \ldots, \ERCW{e_{n-1}}$. See Figure \ref{fig:wOnRight}.

		\begin{figure}
		    \centering 
		    \includegraphics[scale=0.5]{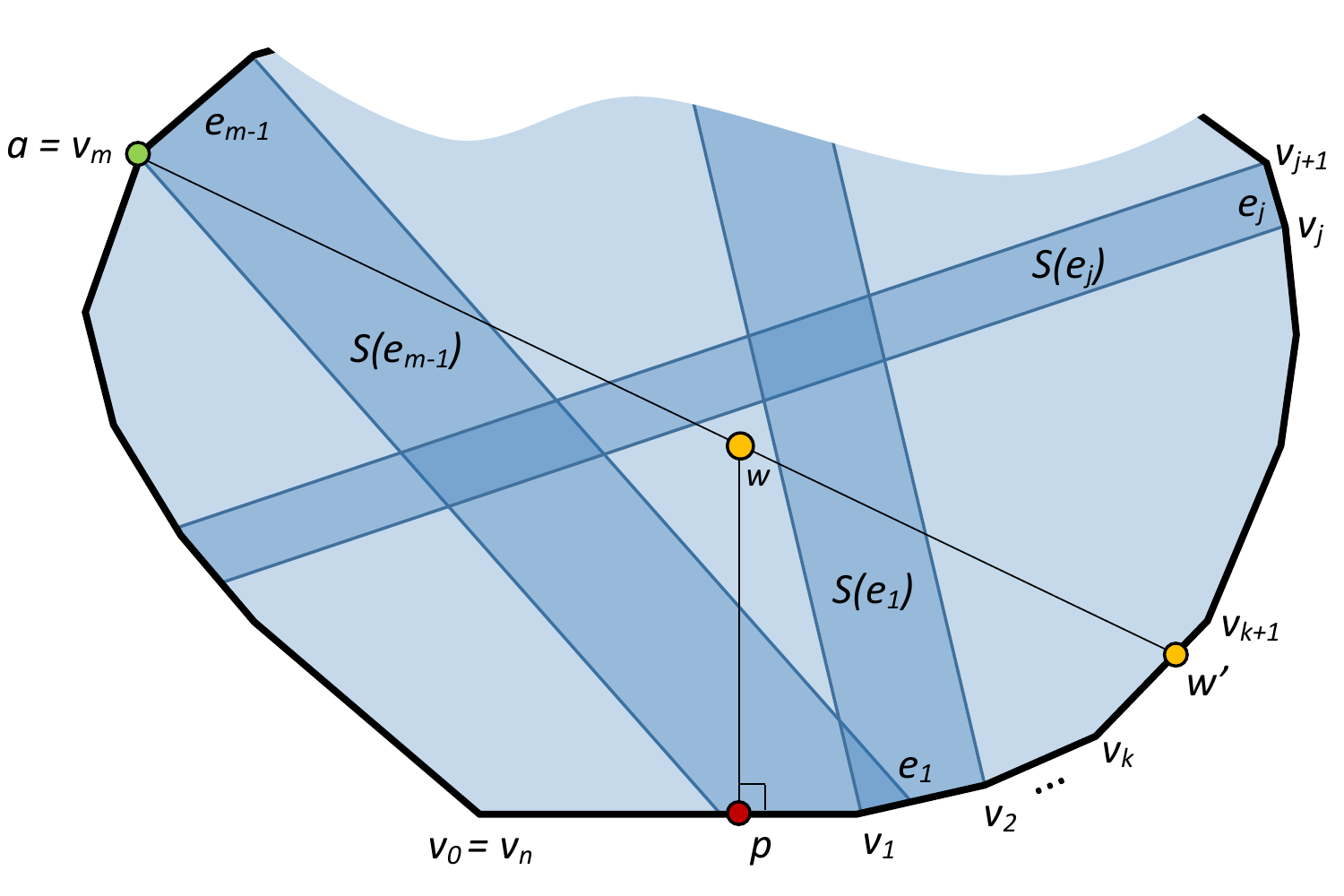}
			\caption{$w$ is in the regions $\ERCCW{e_0}$ to $\ERCCW{e_{m-1}}$.}
			\label{fig:wOnLeft} 
		\end{figure}

		\begin{figure}
			\centering 
			\includegraphics[scale=0.5]{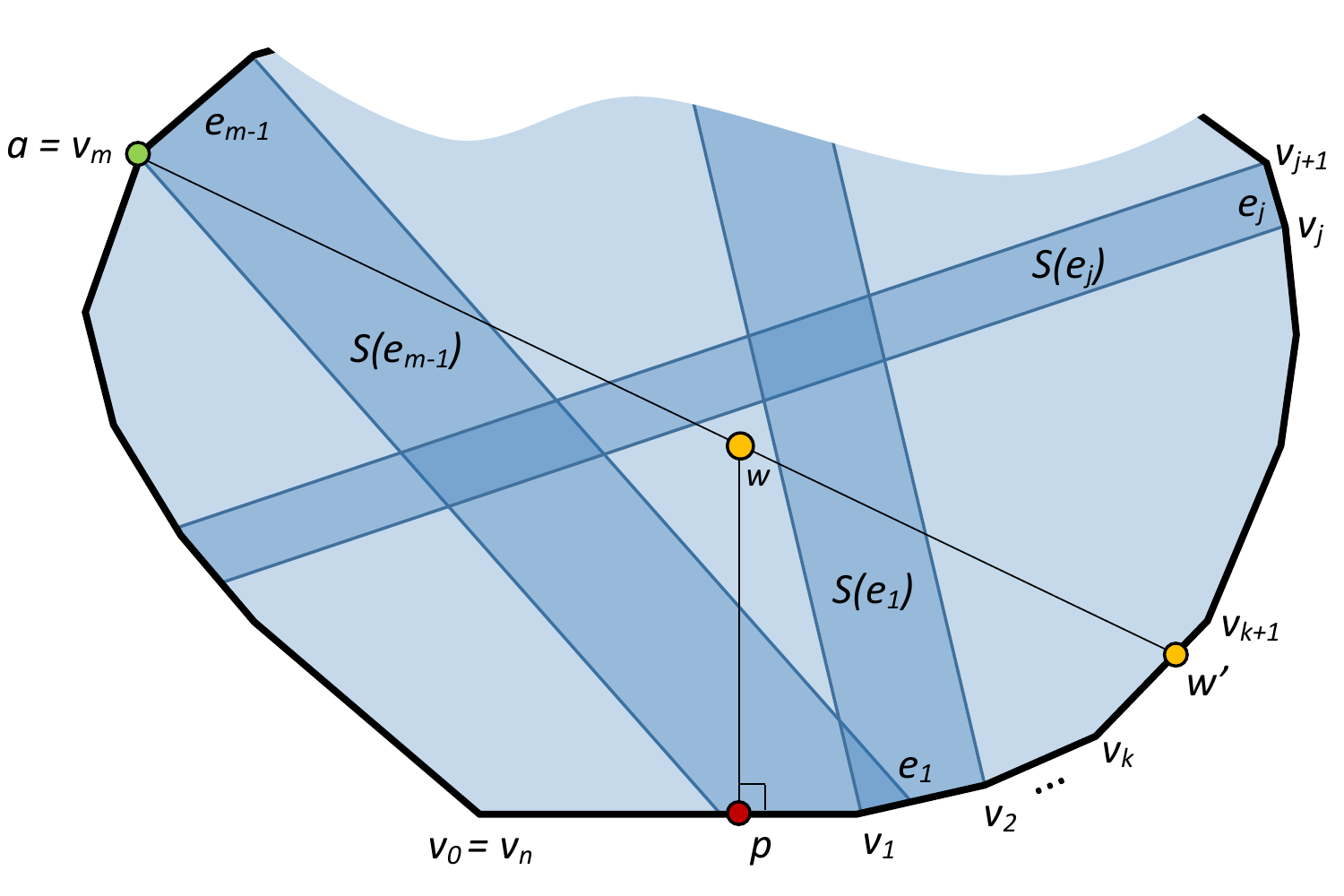}
			\caption{$w$ is in the regions $\ERCW{e_m}$ to $\ERCW{e_{n}}$.}
			\label{fig:wOnRight} 
		\end{figure}
		
		Since all of the slabs $\ERSPLIT{e_m}, \ERSPLIT{e_{m+1}}, \ldots,
		\ERSPLIT{e_{n-1}}$ cross the chord $aw'$ between $a$ and $w$, we have that $w'$ is also in $\ERCW{e_m}$, $\ERCW{e_{m+1}}, \ldots, \ERCW{e_{n-1}}$. Thus, the vertices $v_{m+1}, \ldots, v_n$ move in a clockwise direction to $a$.

		Now, we must show that the particles on vertices $v_1, \ldots, v_{m-1}$ also move to $a$. We first consider the vertices $v_{k+1}, \ldots, v_{m-1}$. Again, since these vertices move counterclockwise when the actuator is activated at $w$, the slabs $\ERSPLIT{e_j}$ for $k \leq j \leq m-1$ cross
		the chord $aw'$ between $a$ and $w$. Therefore, none of them can contain $w'$. This implies that $w'$ is in $\ERCCW{e_{k+1}}, \ldots, \ERCCW{e_{m-1}}$.
		
		We now show that the vertices $v_1, \ldots, v_k$ move in a clockwise direction to $a$.
		Consider the circle $C$ centered at $w$ and
		going through $w'$.
		This circle contains $\bdP[v_1,v_k]$ since particles on $v_1$, $v_2, \ldots, v_k$ move in a counterclockwise direction to $a$ when an actuator is activated at $w$.
		It is strict containment as the particles always move away from $w$.
		
		\begin{figure}
			\centering 
			\includegraphics[scale=0.6]{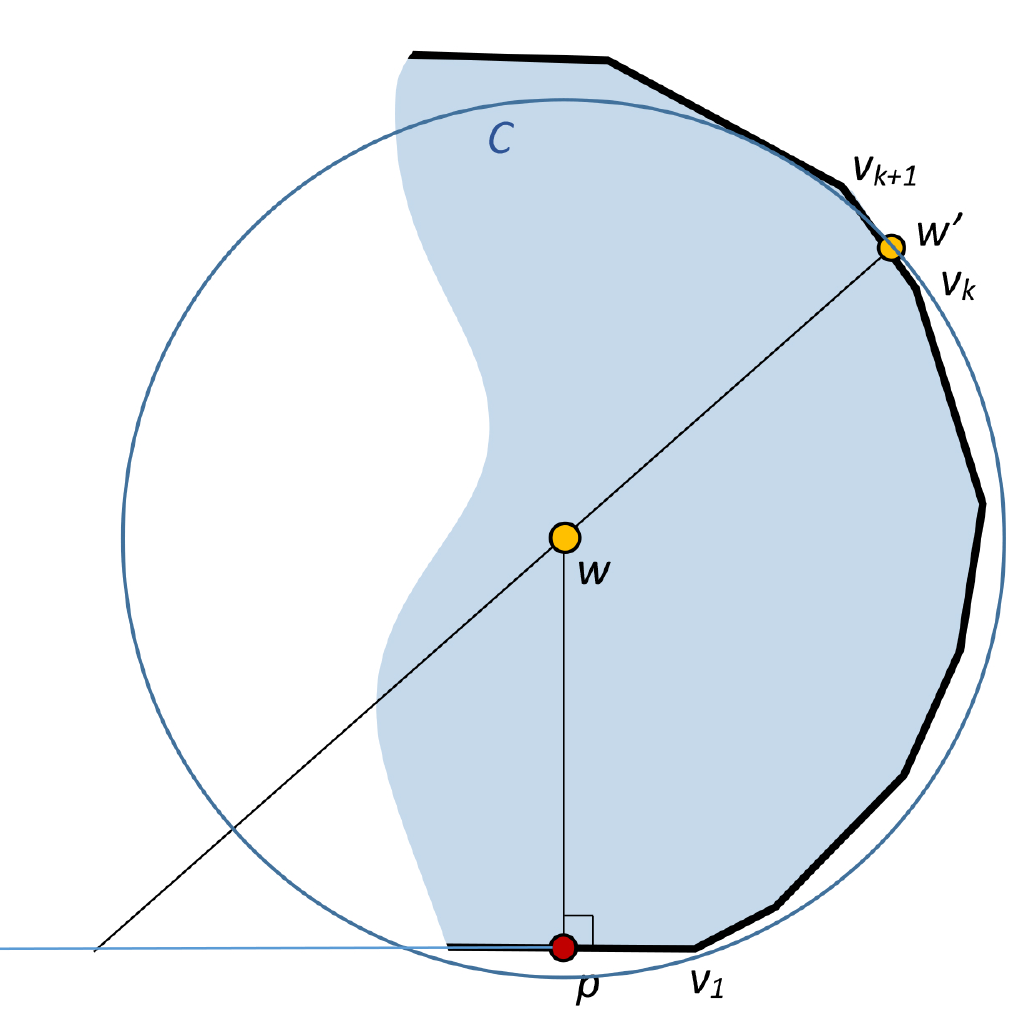}
			\caption{The circle $C$ contains the boundary from $p$ to $w'$.}
			\label{fig:CircleC} 
		\end{figure}
		
		Now consider the circle $C'$ that has the chord $aw'$ as diameter.
		Since $a$ is the accumulation point for $w$, it is the farthest point from $w$.
		This implies that the
		the center $c$ of $C'$ lies on the segment $aw$, with radius
		 $|cw'|$.  $C'$ contains $C$ since $|cw'| > |ww'|$. (Figure \ref{fig:CircleCPrime}).

		\begin{figure}
 			\centering 
 			\includegraphics[scale=0.6]{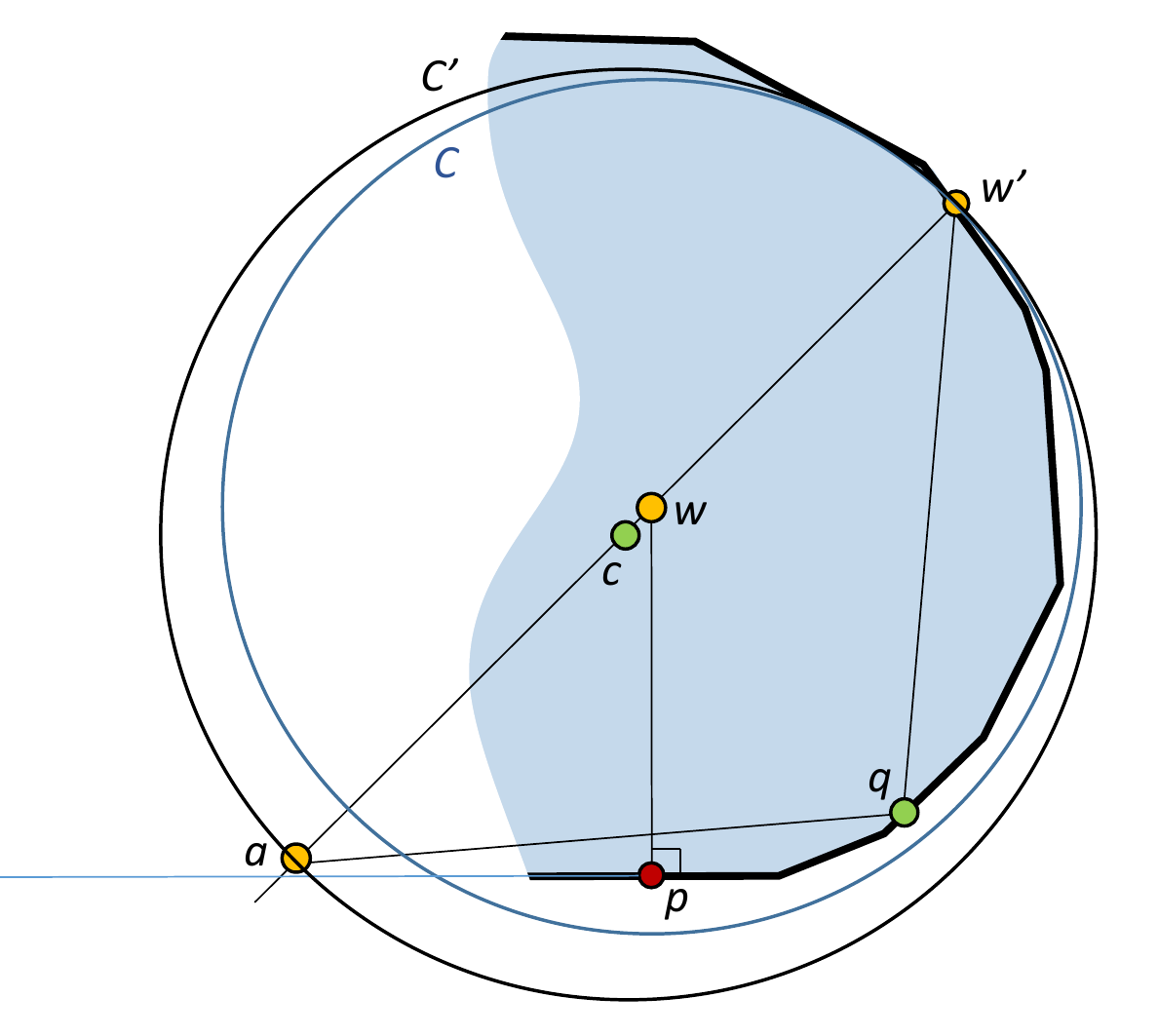}
			\caption{The circle $C'$ contains the circle $C$ and thus contains the
				boundary from $p$ to $w'$.}
			\label{fig:CircleCPrime} 
		\end{figure}
		
		Let $q$ be an arbitrary point in $\bdP(p,w')$. Since $q$ is in the interior of $C$, 
		we have that $\angle{w'qa} > \pi/2$. By convexity, we have that $\angle{w'qp} > \angle{w'qa}$. Consider the cone formed by the ray from $q$ to $w'$ and the ray at $q$ that is an extension of the line through $a$ and $q$. Since $\angle{w'qp} > \pi/2$, we have that the angle formed at this ray is strictly less than $\pi/2$ and $\bdP[q,w')$ is contained in the cone. Lemma 3 in \cite{icking} states that when $\bdP[q,w')$ is contained in a cone at $q$ with angle at most $\pi/2$ for every $q\in \bdP(p,w')$ then $\bdP(p,w')$ is self-approaching from $p$ to $w'$. 
By Lemma \ref{lem:selfapproach}, we have that an activation of an actuator at $w'$ sends $q$ clockwise around the boundary to $p$ since $|pw'| \geq |qw'|$. Therefore, the vertices $v_1, \ldots, v_k$ move in a clockwise direction to $a$.
	
		We have now shown the polygon is 1-gatherable from $w'$.
	
	\end{proof}
		
	As a consequence of the previous lemma, in order to tell if a polygon is
	$1$-gatherable, it suffices to determine if it is $1$-gatherable from the boundary.
	To do this in linear time, we employ an approach that resembles the rotating calipers algorithm to compute the diameter of a convex polygon \cite{shamos}. In essence, for every point $x$ on $\bdP$, we want to compute the first clockwise and first counterclockwise accumulation point. We do this in two steps. We compute all the counterclockwise accumulation points then compute the clockwise accumulation points. The algorithm to compute the counterclockwise accumulation points proceeds as follows. We start at the lowest point $x$ of $P$ and place the first horizontal caliper at $x$. We then walk around the boundary in counterclockwise direction until we find the counterclockwise accumulation point $y$ for $x$. We place the second caliper at $y$ such that it is perpendicular to $xy$. As $x$ moves counterclockwise around $P$, there are two types of events. Either $x$ moves to a new vertex or the caliper at $y$ becomes coincident to an edge of $P$ in which case $y$ moves from one vertex to the next. There are a linear number of events that occur and by recording these events, when the calipers returns to its starting positions, we know the counterclockwise accumulation point for every point on the boundary of $P$. By repeating this in the clockwise direction, we find the clockwise accumulation points. For any point on the boundary of $P$, if its clockwise accumulation point is the same as its counterclockwise accumulation point, then the polygon is 1-gatherable from that point.
We conclude this section with the following:
	\begin{theorem}
	We can determine if a convex $n$-vertex polygon is 1-gatherable in $O(n)$ time.
	\end{theorem}
	\begin{proof}
	Follows from Lemma \ref{lem:boundaryWoof} and the discussion above.
	\end{proof}

\section{2-Gatherability}
   In this section we prove that a convex polygon with at most two acute vertices is 2-gatherable. We then give an $O(n)$ algorithm to determine the location of the two actuators and the sequence of activation.
	
				\begin{figure}
			\centering 
		    \includegraphics[scale=1]{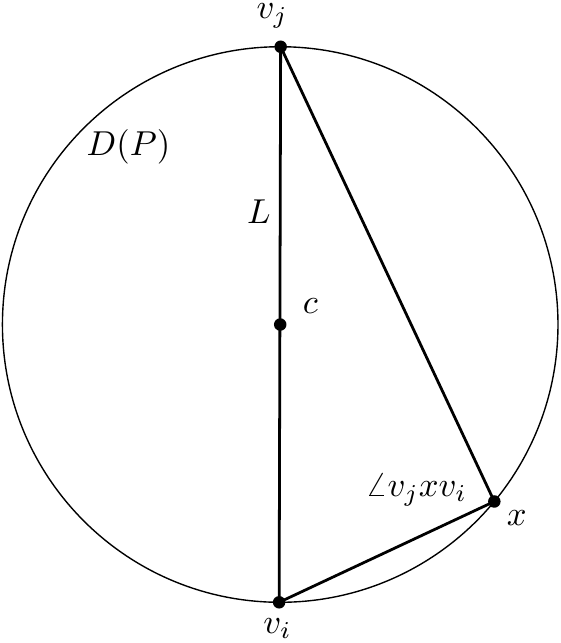}
			\caption{  
				Figure for Lemma \ref{thm:twoWoof}
				}
			\label{fig:2woofig} 
		\end{figure}

		\begin{theorem}\label{thm:twoWoof}
			If a convex polygon has two or fewer acute vertices, then it is 2-gatherable.
		\end{theorem}
		
		\begin{proof}
		
			Let \sed{P} be the smallest disk enclosing polygon $P$ with centre $c$. Either there are two vertices $v_i$ and $v_j$ of $P$ that form a diameter of \sed{P} or there are three vertices $v_i$, $v_j$, and $v_k$ on \bd{P} such that $c$ is in the interior of the triangle formed by the three vertices~\cite{sec1, sec2}. We consider each case separately. Recall that by Lemma \ref{lem:justboundary}, we can assume that the particles are only located on the boundary of $P$.
			
			{\bf Case 1:} Two vertices $v_i$ and $v_j$ of $P$ form a diameter of \sed{P}. In this case, we show that an actuator activated at vertex $v_i$ results in all particles accumulating at $v_j$. Assume, without loss of generality, that $v_i$ and $v_j$ lie on a vertical line $L$ with $v_i$ below $v_j$. The two vertices partition the polygon boundary into two chains, $\bdP[v_i,v_j]$ which is to the right of $L$ and $\bdP[v_j,v_i]$ which is to the left. We also assume that each chain consists of at least two edges, since otherwise, one of the chains is the edge $v_iv_j$ and trivially any particle on this edge moves to $v_j$ when an actuator at $v_i$ is activated. To complete the proof in this case, by Lemma \ref{lem:selfapproach}, it suffices to show that both $\bdP[v_i,v_j]$ and $\bdP[v_j,v_i]$ are self-approaching curves from $v_j$ to $v_i$. 
			
			Consider any point $x\in \bdP(v_i,v_j)$. Since $x$ is in \sed{P} strictly to the right of $L$ we have that $\pi > \angle{v_jxv_i} \geq \pi/2$. Consider the cone formed by the intersection of the half-space bounded by the line through $v_j$ and $x$ that contains $v_i$ and the half-space bounded by the line through $v_i$ and $x$ that does not contain $v_j$. This cone has angle at most $\pi/2$ and contains $\bdP[v_i,x]$. Since $x$ is an arbitrary point on $\bdP(v_i,v_j)$, by Lemma 3 in \cite{icking}, we have that $\bdP[v_i,v_j]$  is self-approaching from $v_j$ to $v_i$. A similar argument shows that $\bdP[v_j,v_i]$  is also self-approaching from $v_j$ to $v_i$.
			
			{\bf Case 2:} There are three vertices $v_i$, $v_j$, and $v_k$ appearing in counter-clockwise order on \bd{P} such that $c$ is in the interior of the triangle formed by the three vertices. Since there are at most two acute vertices, without loss of generality, assume that $v_j$ is a polygon vertex with interior angle at least $\pi/2$. Reorient the polygon such that $v_i$ is the lowest point. The polygonal chains $\bdP[v_i,v_j]$, $\bdP[v_j,v_k]$ and $\bdP[v_k,v_i]$ are self-approaching from $v_j$ to $v_i$, $v_k$ to $v_j$ and $v_k$ to $v_i$, respectively, by the same argument as the one used in Case 1. In fact, since $c$ is strictly in the interior of the triangle formed by the three vertices, we have that the cones used to prove that the chains are self-approaching have an angle that is strictly less than $\pi/2$. 
			
			By placing a first active actuator on $v_i$, we have that all the particles on $\bdP(v_i,v_j]$ and all the particles on $\bdP[v_k,v_i)$ move  onto $\bdP[v_j,v_k]$. Since $\bdP[v_j,v_k]$ is self-approaching from $v_k$ to $v_j$, if we activated a second actuator at $v_j$ then all the particles on this chain move to $v_j$'s accumulation point which would complete the proof. However, even though $v_j$ is not acute, it may be the case that $v_j$ is the counterclockwise accumulation point for $v_i$. This would prevent us from placing an actuator on $v_j$ since after the activation of the first actuator on $v_i$, particles have accumulated on $v_j$. Recall that all subsequent placements of actuators must be on points in $P$ that are free of particles. Since for every point $x$ on $\bdP(v_j,v_k)$, $\angle{v_jxv_k} > \pi/2$, there must exist a point $y$ on the edge $v_jv_{j-1}$ infinitessimally close to $v_j$ such that the $\angle{yzv_k}$ is still strictly greather than $\pi/2$ for every $z\in \bdP[v_j,v_k)$. This implies that $\bdP[y,v_k]$ is self-approaching from $v_k$ to $y$. Thus, by Lemma \ref{lem:selfapproach},  activating a second actuator at $y$, which is free of particles after the first activation, moves all the particles that have accumulated on $\bdP[v_j,v_k]$ to the counterclockwise accumulation point of $y$.

		\end{proof}

\bibliography{ref}

\end{document}